\documentclass[11pt]{article}
\usepackage[margin=1in]{geometry}
\usepackage{amsmath,amsthm,amssymb}
\usepackage{xcolor}
\usepackage{hyperref}
\usepackage[nameinlink,capitalise]{cleveref}
\usepackage{bbold}
\usepackage{mathtools}
\usepackage{doi}
\usepackage{tikz}

\newtheorem{theorem}{Theorem}[section]
\newtheorem{lemma}[theorem]{Lemma}
\newtheorem{conjecture}[theorem]{Conjecture}

\newtheorem{proposition}[theorem]{Proposition}
\theoremstyle{definition}
\newtheorem{definition}[theorem]{Definition}

\newcommand{\doublehat}[1]{%
    \hat{\hat{#1}}
}

\renewcommand{\d}{\,\mathrm{d}}
\newcommand{\h}{\hat}

\newcommand{\eps}{\varepsilon}

\definecolor{myblue}{rgb}{0.15, 0.1, 0.95}
\definecolor{mygreen}{rgb}{0.15, 0.65, 0.25}
\definecolor{myred}{rgb}{0.75, 0.25, 0.15}

\hypersetup{
    colorlinks=true,
    linkcolor=myred,
    citecolor = myblue,      
}

\title{Auditing for Core Stability in Participatory Budgeting\thanks{Supported by NSF grant CCF-2113798.}}

\author{Kamesh Munagala\thanks{Department of Computer Science, Duke University, Durham, NC 27708-0129.  Email: \texttt{kamesh@cs.duke.edu}.} \and Yiheng Shen\thanks{Department of Computer Science, Duke University, Durham, NC 27708-0129.  Email:  \texttt{yiheng.shen@duke.edu}.} \and Kangning Wang\thanks{Simons Institute for the Theory of Computing, UC Berkeley, Berkeley, CA 94720-2190. Email: \texttt{kangning@berkeley.edu}. This work was done while the author was at Duke University.}
}

\date{}

\begin{document}
\maketitle

\begin{abstract}
We consider the \emph{participatory budgeting} problem where each of $n$ voters specifies additive utilities over $m$ candidate projects with given sizes, and the goal is to choose a subset of projects (i.e., a \emph{committee}) with total size at most $k$. Participatory budgeting mathematically generalizes multiwinner elections, and both have received great attention in computational social choice recently. A well-studied notion of group fairness in this setting is \emph{core stability}: Each voter is assigned an ``entitlement'' of $\frac{k}{n}$, so that a subset $S$ of voters can pay for a committee of size at most $|S| \cdot \frac{k}{n}$. A given committee is in the core if no subset of voters can pay for another committee that provides each of them strictly larger utility. This provides proportional representation to all voters in a strong sense.
In this paper, we study the following auditing question: Given a committee computed by some preference aggregation method, how close is it to the core? Concretely, how much does the entitlement of each voter need to be scaled down by, so that the core property subsequently holds? 
As our main contribution, we present computational hardness results for this problem, as well as a logarithmic approximation algorithm via linear program rounding. We show that our analysis is tight against the linear programming bound. Additionally, we consider two related notions of group fairness that have similar audit properties. The first is \emph{Lindahl priceability}, which audits the closeness of a committee to a market clearing solution. We show that this is related to the linear programming relaxation of auditing the core, leading to efficient exact and approximation algorithms for auditing. The second is a novel weakening of the core that we term the \emph{sub-core}, and we present computational results for auditing this notion as well.

\end{abstract}

\section{Introduction}
The \emph{participatory budgeting} problem~\cite{cabannes2004participatory,aziz2021participatory,pbstanford,knapsackVoting,DBLP:conf/wine/FainGM16} is motivated by real-world elections where voters decide which projects a city should fund subject to a budget constraint on the total cost of these projects. In this problem, there are $m$ candidate projects forming a set $C$, and $n$ voters. Each candidate $j$ is associated with a size/cost $s_j$.

The \emph{multiwinner election} problem~\cite{AzizChapter,EndrissBook,VinceBook,CC,thiele1895om} is commonly seen in practice, and has received significant research attention recently. Mathematically, it is a specialization of the participatory budgeting problem, where each candidate is of the same unit size.

In both settings, our goal is to pick a subset $T \subseteq C$ of candidates -- which we call a \emph{committee} -- with total size at most a given value $k$, that is, $\sum_{j \in T} s_j \le k$. Each voter $i$ has a utility function $U_i(T)$ over subsets $T \subseteq C$ of candidates. In this paper, we assume the utility functions $\{U_i\}$ are \emph{additive} across candidates. For some of our results, we also look at the more restricted case of multiwinner elections with \emph{approval} (i.e. 0/1-additive) utilities: Each candidate is of unit size; each voter $i$ ``approves'' a subset $A_i \subseteq C$ of candidates, and for any committee $T$, the utility function of voter $i$ is simply $U_i(T) = |T \cap A_i|$, the number of approved candidates in the committee. We call this the {\sc Approval Election} setting.

\paragraph{Core Stability.} In both multiwinner elections and participatory budgeting, the methods used to aggregate preferences of voters are typically very simple, for instance, choosing the candidates who receive the most approval votes. This leads to a tension of such rules with \emph{fairness} of the resulting outcome in terms of proportional representation of minority opinions, and a social planner may want to quantify this tension for any given election.

A notion of fairness in this context, which has been studied for over a century, is that of \emph{core stability}~\cite{scarf1967core,DBLP:conf/ec/FainMS18,Droop,thiele1895om,lindahl1958just}. This captures a strong notion of proportional representation. Given a committee $W$ of size $k$, think of $k$ as a budget, and split it equally among the $n$ voters, so that each voter is entitled to a budget of $\frac{k}{n}$. For any subset $S \subseteq [n]$ of voters, their total entitlement is $|S|\cdot \frac{k}{n}$. If there is another committee $T$ of size at most the entitlement $|S| \cdot \frac{k}{n}$, such that each voter $i \in S$ strictly prefers $T$ to $W$, {\em i.e.}, $U_i(T) > U_i(W)$ for all $i \in S$, then these voters would have a justified complaint with $W$. A committee $W$ where no subset $S \subseteq [n]$ of voters have a justified complaint is termed {\em core stable}.

The core has a ``fair taxation'' interpretation~\cite{lindahl1958just,foley1970lindahl}. The quantity $\frac{k}{n}$ can be thought of as the tax contribution of a voter, and a committee in the core has the property that no sub-group of voters could have spent their share of tax money in a way that \emph{all} of them were better off. It subsumes notions of fairness such as Pareto-optimality, proportionality, and various forms of justified representation~\cite{JR,Sanchez,PJR2018} that have been extensively studied in multiwinner election and fairness literature.  Note that the core is {\em oblivious} to how the demographic slices are defined -- it attempts to be fair to {\em all} subsets of voters. This is a desirable feature in practice, since demographic slices are often not known upfront, and there could be hidden sub-groups that can only be inferred from voter preferences.

\paragraph{Approximate Stability.} The core is a very appealing group fairness notion; however, even in very simple settings, the core could be empty~\cite{DBLP:conf/ec/FainMS18}. This motivates  approximation, where the entitlement $\frac{k}{n}$ of each voter is scaled by a factor of $\theta$.

\begin{definition}
\label{def:approx_core}
For $\theta \le 1$, a committee $W$ of size at most $k$ lies in the \emph{$\theta$-approximate core} if for all $S \subseteq [n]$, there is no \emph{deviating committee} $T$ with size at most $\theta\cdot |S|\cdot \frac{k}{n}$, such that for all $i \in S$, we have $U_i(T) > U_i(W)$.
\end{definition}

It is known~\cite{DBLP:conf/stoc/JiangMW20} that a $\frac{1}{32}$-approximate core solution always exists for very general utility functions of the voters.  

\paragraph{Auditing for Approximate Stability.} Though the existence of approximate core solutions is a strong positive result, the algorithms for finding such solutions are often complex. Indeed, even in settings where the core is known to be always non-empty, for instance when candidates can be chosen fractionally~\cite{foley1970lindahl}, the non-emptiness is an existence result that needs an expensive fixed point computation. On the other hand, in practice, what are implemented are typically the simplest and most explainable social choice methods such as 
Single Transferable Vote (STV). Therefore, from the perspective of a societal decision maker, such as a civic body running a participatory budgeting election, it becomes important to answer the following {\em auditing} question for any given election:

\begin{quote}
Given a committee $W$ of size at most $k$ found by some implemented preference aggregation method, how close is it to being core stable, i.e., what is the smallest value of $\theta_c$ such that $W$ {\em does not} lie in the $\theta_c$-approximate core for that instance?
\end{quote}

Note that if a committee $W$ lies in the core, then $\theta_c > 1$, else $\theta_c \le 1$.  Such an auditing question is useful even if the decision maker themselves is not sensitive for fairness because it allows for review of implemented decision rules via a third party or government agency. Further, the  set of deviating voters that correspond to the $\theta_c$-approximation yield a demographic that are unhappy with the current outcome, and this can be analyzed further by policy makers. 

We term the above question as the {\em core auditing} problem. In this paper, we study the computational complexity of core auditing. In that process, we define both stronger and weaker notions of fairness and audit these notions as well.

\subsection{Our Results}
\paragraph{Hardness and Approximation Algorithm.} We show in \cref{sec:hard} that for {\sc Approval Elections}, the value of $\theta_c$ in the core auditing problem is {\sc NP-Hard} to approximate to a factor better than $1 + \frac{1}{e} > 1.367$. We further show that this {\sc APX-Hardness} persists even when voters are allowed to choose a fractional deviating committee. We also show that the problem remains {\sc NP-Hard} when each voter approves a constant number of candidates, and each candidate is approved by a constant number of voters. These results significantly strengthen the {\sc NP-Hardness} result presented in~\cite{BrillGPSW20}.

On the positive side, in \cref{sec:approx1}, we design an $O(\min(\log m, \log n))$ approximation algorithm for the value $\theta_c$, where $m$ and $n$ are the number of candidates and voters respectively. We do this via linear program rounding. Our program (and indeed, our auditing question itself) is an interesting generalization of the densest subgraph problem~\cite{Charikar-dense}, where the goal is to choose a subgraph with maximum average degree. Given a graph, treat voters as edges and candidates as vertices that are approved by the incident edges; further assume any voter needs utility $2$ (that is, both end-points) in a feasible deviation. Then, the value of $\theta_c$ is precisely the density of the densest subgraph (to scaling). We combine ideas from the rounding for densest subgraph (where the rounding produces the integer optimum without approximation) with that from maximum coverage to design our rounding scheme. 
We further show that our linear program has an integrality gap of $\Omega(\min(\log m, \log n))$, showing that we cannot do any better against an LP lower bound. Our proof in \cref{sec:approx1} applies to the {\sc Approval Election} setting. We extend this to general candidate sizes and arbitrary additive utilities via knapsack cover inequalities in \cref{sec:general}, leading to an $O(\min(\log m, \log n))$ approximation factor. We finally note that both our hardness (see \cref{app:hard2}) and approximation results easily extend to settings where candidates can be fractionally chosen in the committees. 

It is an interesting question to close the gap between our hardness result (constant factor) and our approximation ratio. The difficulty lies with density problems in general, where hardness of approximation results have been hard to come by; see for instance, the $k$-densest subgraph problem~\cite{Khot-densest}.

\paragraph{Lindahl Priceability.} A closely related notion of fairness, considered in~\cite{lindahl1958just,foley1970lindahl,MunagalaSWW22,Peters21b} is that of committees that can be supported by market clearing prices. The notion of {\em Lindahl equilibrium} is a pricing scheme that strengthens the core, meaning that if the former exists, it lies in the core. In this scheme, each voter $i$ is assigned price $p_{ij}$ for candidate $j$, and these prices are such that for any candidate, the total price is equal to its size. If a voter buys their optimal set of candidates subject to the total price paid being at most their entitlement, $k/n$, then all voters choose the {\em same} committee. This is therefore a market clearing notion with per-voter prices such that the optimal voter action given these prices and equal entitlements results in a common committee being chosen. If committees could be chosen fractionally, it is known via a fixed point argument that the Lindahl equilibrium always exists~\cite{foley1970lindahl}. However, these need not exist when considering integer committees.


In this paper, we consider an integer version of this concept that we term {\em Lindahl priceability}. We show that this notion implies the core. As with the core, in \cref{sec:comp}, we define the approximation factor $\theta_{\ell}$ to which a given committee  satisfies Lindahl priceability, via scaling the entitlement $k/n$ of each voter by that factor. We show via LP duality not only that the quantity $\theta_{\ell}$ can be audited in polynomial time for {\sc approval elections}, but also that this computation coincides with the LP relaxation to the core auditing program. This results in a novel and somewhat surprising connection between the Lindahl priceability and the core for {\sc Approval elections}, where the approximation factor $\theta_{\ell}$ for Lindahl priceability is found via the LP relaxation to the program that computes the approximation factor $\theta_c$ for the core. Further, our approach easily extends to show computational results for general utilities and sizes.

Our notion is related to the {\em cost efficient Lindahl equilibrium} proposed recently in~\cite{Peters21b} for {\sc Approval elections}. However, there is a crucial difference: While they translate the fractional Lindahl equilibrium to the integer case, we translate the gradient optimality conditions implied by the fractional  equilibrium to the integer case. To illustrate that our definition is different, note that while there are simple instances of {\sc Approval elections} on which the former notion does not exist, we do not know such an instance for our definition. 

\paragraph{Weak Priceability and Sub-core.} In \cref{sec:subcore}, we finally connect our work to another notion of priceability first studied in~\cite{peters2020proportionality}. This notion is a relaxation of Lindahl priceability for {\sc Approval elections}, where voters cannot greedily augment the current committee given the prices and their entitlement. We term this ``weak priceability'' and use this to define a new relaxation of the core, termed {\em sub-core}, which only allows voters to deviate and gain utility from super-sets. We show that weakly priceable committees lie in the sub-core. Further, though the sub-core appears like a weak notion of fairness, we show that it remains {\sc NP-Hard} to audit. We finally present an $O(\min(\log m, \log n))$ approximation to the auditing question using same techniques as for auditing the core. 

In practice, committees found by social choice rules are likely to be much better approximations to the sub-core compared to the core. Hence, it is desirable to show a practitioner closeness to weaker notions of fairness such as the sub-core in addition to closeness to the core.

\subsection{Related Work}
\paragraph{Proportionality in Social Choice.} The earliest work that considers proportional representation dates back to the late 1800's~\cite{Droop}, and several voting rules attempting to achieve it, such as PAV~\cite{thiele1895om} and Phragm\'en~\cite{phragmen} rules also date back to then. There has been resurgence of interest in axiomatizing proportionality~\cite{CC,Monroe,Brams2007,JR,Sanchez,PJR2018} partly driven by real-world applications of such elections to areas such as participatory budgeting~\cite{pbstanford,knapsackVoting,aziz2021participatory}, and partly due to local bodies and countries implementing rules such as ranked choice voting that attempt to achieve proportionality, in their elections. These advances have made auditing fairness notions such as closeness to the core and weaker group fairness notions imperative in these settings.

\paragraph{Notions of Approximate Core.} In addition to the notion of approximation presented in \cref{def:approx_core}, a different notion allows deviating voters to use their entire entitlement, but requires them to extract at least a factor $\theta > 1$ larger utility on deviation. Under this notion, it is shown in~\cite{peters2020proportionality} that a classic voting rule called PAV~\cite{thiele1895om} achieves a $2$-approximation. This result was generalized to show a constant approximate core for arbitrary submodular utility functions and general candidate sizes in~\cite{MunagalaSWW22}. An analogous result for clustering was presented in~\cite{ChenFLM19}. Our work directly shows that this notion of approximation can be audited in a {\em bicriteria} fashion as follows: If the given committee is a $c$-approximation without violating entitlements, we can determine if it is a $c$-approximation had entitlements been violated by a factor of $O(\log m)$. It is an interesting open question to remove the bicriteria nature of this result.

\paragraph{Auditing for Fairness.} The question of auditing has become salient given the increasing democratization of societal decision making, for instance via processes like participatory budgeting. In the context of social choice, there are natural properties that are easy to achieve algorithmically but hard to audit. For instance, checking if an arbitrary outcome is Pareto-optimal is computationally hard~\cite{Aziz16}, while achieving it via some algorithm is easy. We take a further step in this direction by studying the approximate audit of arguably the strongest possible group fairness notion, the core, as well as related fairness properties.

Going beyond social choice, the notion of auditing for group fairness has gained prevalence in machine learning. Here, the ``voters'' are data points, and the ``committee'' is a classifier. We wish to audit if the classifier provides comparable accuracy for various demographic slices. The work of~\cite{Kearns} formulates and presents algorithms for this problem.

\section{Mathematical Program for $\theta_c$}
\label{sec:program}
For most of this paper, we consider the {\sc Approval Election} setting. Recall that in this setting, each voter $i$ ``approves'' a set $A_i \subseteq C$ of unit-sized candidates, and its utility for a committee $T \subseteq C$ is simply $U_i(T) = |A_i \cap T|$. Our hardness results hold even for this simple setting, while our approximation algorithms hold for general additive utilities and sizes (see \cref{sec:general} and \cref{sec:lindahl_fractional}).

We first present a mathematical program that computes $\theta_c$ given a committee $W \subseteq C$ of size at most $k$, as in \cref{def:approx_core}. In this program, there is a variable $z_i \in \{0,1\}$ that captures whether voter $i$ deviates, and a variable $x_j \in \{0,1\}$ that captures whether candidate $j$ is present in the deviating committee. If this is a feasible deviation, then the utility of each voter for which $z_i =1$ must strictly increase, which means 
$
\forall i \in [n],\ \sum_{j \in A_i} x_j \ge (U_i(W) + 1)\cdot z_i.
$

Next, let $R = \frac{n}{k}$. Then, the budget available to the deviating voters is $\frac{1}{R} \sum_i z_i$, while the size of the committee to which they deviate is $\sum_{j\in C} x_j$. This means the entitlement $k/n$ of each voter must be scaled by a factor of $
R \cdot \frac{\sum_{j\in C} x_j}{\sum_i z_i},
$
so that the voters with $z_i = 1$ do not have enough entitlement to pay for this deviating committee. Since the goal is to have no deviations at all, the value $\theta_c$ is simply the solution to the following mathematical program:
\begin{gather*}
    \text{Minimize $R\cdot \frac{\sum_{j\in C} x_j}{\sum_i z_i}$},\ \text{s.t.}\\
    \forall i\in [n],\ \sum_{j \in C \cap A_i} x_j \ge z_i \cdot (U_i(W) + 1);\\
    \forall i\in [n],\ \forall j\in [m],\ x_j, z_i \in \{0,1\}.
\end{gather*}

The above program attempts to maximize the ratio of the number of constraints satisfied via setting $z_i$ to $1$, to the number of $x_j$ variables set to $1$. 

\section{Hardness of Auditing the Core}
\label{sec:hard}

As mentioned before, all hardness results in this section apply to the {\sc Approval election} setting, where the utilities are binary, and candidate sizes are unit. We first show that the core auditing problem, that is, the problem of computing $\theta_c$ for a given committee $W$, is {\sc NP-Hard} even in a ``constant degree'' setting. This strengthens an {\sc NP-Hardness} result for the core in~\cite{BrillGPSW20}.

\begin{theorem} [Proved in \cref{app:hard}]
\label{thm:hard1}
Deciding whether a committee $W$ does not lie in the core (that is, deciding whether its $\theta_c \le 1$) is {\sc NP-Hard} when each voter approves at most $6$ candidates (that is, $|A_i| \le 6$ for all voters $i \in [n]$), and each candidate lies in at most $2$ of the sets $A_i$.
\end{theorem}

We now show that the core auditing problem is in fact {\sc APX-Hard}.

\begin{theorem}\label{thm:standard_inapprox}
For any constant $\gamma > 0$, approximating $\theta_c$ to within a factor of $1 + \frac{1}{e} - \gamma$ is {\sc NP-Hard}.
\end{theorem}
We will reduce from the maximum set coverage problem on regular instances.
\begin{lemma} [Regular Maximum Coverage~\cite{feige1998threshold}]\label{lem:regular_cover}
The universe contains $qd$ elements. There are $\xi$ sets, each with $d$ elements. It is {\sc NP-Hard} to distinguish between the following two cases:
\begin{itemize}
	\item [1.] ``YES'' instances: There exist $q$ sets that cover the universe.
	\item [2.] ``NO'' instances: No collection of $q$ sets can cover $(1-1/e+\eps)\cdot qd$ elements.
\end{itemize}  
\end{lemma}
\begin{proof}[Proof of \cref{thm:standard_inapprox}]
    For each instance of the regular Max Covering Problem, there are $q d$ elements and $\xi$ sets. We construct the following instance for auditing the core:
    \begin{itemize}
        \item There are $\xi$ main candidates. Each candidate corresponds to a set. There are $\frac{1}{e} (q-1)q  d$ dummy candidates. 
        \item There are two group of voters. The first group contains $\frac{1}{e}\cdot q d$ voters. They each approve $q-1$ disjoint dummy candidates, and all the main candidates. 
        \item The second group contains $q d$ voters. Each of these voters corresponds to an element of the covering instance. She approves the main candidates whose corresponding set contains her corresponding element. Therefore, there are $(1+1/e) qd$ voters. Add dummy voters who do not approve any candidates, so that the total number of voters is $n = q(q-1)d^2$.
        \item The budget for committee selection is $k = (q-1)q d$. The current committee $W$ contains all the dummy candidates. All voters in the first group have utility $q-1$ while all voters in the second group have utility $0$ in $W$.
        \item Note that each voter is assigned a budget of $\frac{1}{R} = \frac{k}{n} = \frac{(q-1)qd}{q(q-1)d^2} = \frac{1}{d}$. 
    \end{itemize}

    If the maximum coverage instance is a ``YES'' instance, choose as the deviating committee the $q$ main candidates whose corresponding sets cover the universe. We call a voter ``satisfied'' if  her utility has strictly increased compared to the current committee $W$. From the program in \cref{sec:program}, $\theta_c$ is $R = d$ times the minimum ratio of the total number of selected candidates to the number of satisfied voters. Since we have selected $q$ candidates, the voters in the first group receive utility $q$ and are  therefore satisfied. Moreover, since the chosen candidates' corresponding sets cover the universe of $qd$ elements, the voters in the second group receive utility at least one, and are therefore satisfied. Therefore, 
    $$\theta_c\le R\cdot \frac{q}{qd\cdot (1+1/e)}=\frac{1}{1+1/e}.$$
    
    Suppose the maximum coverage instance is a ``NO'' instance. We will show that $\theta_c \ge 1-o(1)$. First suppose a deviating committee is composed of $s < q$ main candidates. These candidates can cover at most $d s$ voters from the second group. For the first group, they provide utility $s$ to each voter. If $t$ of these voters are satisfied, we must have chosen $(q-s)t $ dummy candidates. This means the scaling factor needed is at least
    $$ R \cdot  \frac{s + (q-s) t}{d s + t} = \frac{s + (q-s)t}{s + t/d} > 1.$$
    
    If the number of main candidates in the deviating committee is at least $q$, the voters in the first group are all satisfied and we don't need to choose dummy candidates. Consider an arbitrary $q$-candidate subset of these selected candidates. All voters in the first group are satisfied by these candidates, since they receive utility $q$ from them. Since the coverage instance is a ``NO'' instance, no more than $(1-1/e+\eps)\cdot kd$ voters in the second group are satisfied by this subset. Suppose there are $r$ remaining candidates in the deviation, each candidate can only increase the number of satisfied voters by at most $r \cdot d$. Therefore,  
   \[
        \theta_c\ge R\cdot \frac{q+r}{r\cdot d+ (1-1/e+\eps)\cdot q \cdot d+(1/e)\cdot q\cdot d} =\frac{R}{d}\cdot \frac{q+r}{r\cdot d+q\cdot d\cdot (1+\eps)} \ge\frac{1}{1+\eps}. 
   \]
Since the gap of $\theta_c$ between the constructed auditing instance from ``YES''  instances and from ``NO'' instances is at least $\frac{1+1/e}{1+\eps}$, approximating $\theta_c$ to within this factor is {\sc NP-Hard}.
\end{proof}

\paragraph{Hardness of Auditing Fractional Committees.} One natural question is whether the above hardness stems from the integrality requirement on the committee (the $x_j$ variables in the program in \cref{sec:program}) or the voters (the $z_i$ variables). In \cref{app:hard2}, we show that the auditing problem remains hard to approximate to constant factors even when the committees can be chosen {\em fractionally}. This corresponds to allowing the variables $\{x_j\}$ to be fractional in $[0,1]$. This shows that the hardness of the problem stems mainly from insisting $\{z_i\}$ be integral. The proof of this result is similar to the previous proof. 

\section{A Logarithmic Approximation for Auditing the Core}
\label{sec:approx1}
Our main result in this section is the following theorem, which we prove for the {\sc Approval Election} setting. The proof for general candidate sizes and general additive utilities is presented in \cref{sec:general}.

\begin{theorem}\label{thm:log_approx}
Given a committee $W$ of size at most $k$, its $\theta_c$ value can be computed within $O(\min(\log m, \log n))$ factor in polynomial time, where $m,n$ are the total number of candidates and voters respectively.
\end{theorem}


\paragraph{LP Relaxation.} Given a committee $W$, we start with the mathematical program from \cref{sec:program} and relax the variables to be fractional. This yields the following program. To see that this is a relaxation, if $z_i = 0$ for some $i$, then 
the first constraint is trivially satisfied. On the other hand, if $z_i = 1$, then we can increase all $y_{ij}$ so that $y_{ij} = x_j$, thereby recovering the constraint in the integer program from  \cref{sec:program}. Therefore, any solution to the integer program is a feasible solution to the program below. 
\begin{gather*}
    \text{Minimize $R\cdot \frac{\sum_{j=1}^m {x_j}}{\sum_{i=1}^n z_i}$},\ \text{s.t.}\\
    \forall i\in [n],\ \sum_{j\in A_i} y_{ij}\ge z_i\cdot (U_i(W)+1);\\
    \forall i\in [n],\ \forall j\in A_i,\ y_{ij}\le x_j;\\
    \forall i\in [n],\ \forall j\in A_i,\ y_{ij}\le z_i; \\
    \forall i \in [n], j \in [m],\ x_j, z_i, y_{ij} \ge 0.
\end{gather*}
This can be written as a LP if we omit the denominator from the objective and add the constraint $\sum_i z_i \ge 1$, and hence can be solved in polynomial time.

Denote $u_i = U_i(W)$. For the committee $W$, we further denote 
\begin{equation}
    \label{eq:thetap}
    \theta_p = R\cdot \frac{\sum_{j=1}^m {x_j}}{\sum_{i=1}^n z_i}
\end{equation}
where the variables are set based on the optimal solution to the linear relaxation. Therefore, $\theta_p \le \theta_c$. We will now prove \cref{thm:log_approx} by showing that $\theta_p$ is an $O(\log m)$ approximation to $\theta_c$.

\subsection{Proof of \cref{thm:log_approx}}
\label{sec:roundproof}
By scaling the above program, we can assume that $\max_i\{z_i\} = 1$. Therefore, all $y_{ij} \leq z_i \leq 1$ and $x_j = \min_{i: j \in A_i}\{y_{ij}\} \leq 1$: all the variables are in the range $[0,1]$.

\paragraph{$O(\log m)$ Approximation.} Given the fractional solution, we note that $y_{ij} = \min(x_j, z_i)$. We now construct an integral solution  by the following steps:
\begin{itemize}
    \item [1.] Pick $\alpha \in [0,1]$ uniformly at random. If $z_i> \alpha$, set $\hat{z}_i=1$; else $\hat{z}_i=0$.
    \item [2.] Let $x_j'=\max\{\frac{1}{2m^2},x_j\}$.
    \item [3.] If $2x_j'> \alpha$, then set $\hat{x}_j=1$; else set $\hat x_j=1$ with probability $2x_j'/\alpha$. We round each $\hat{x}_j$ independently.
    \item [4.] If $\hat z_i=1$, check if $\sum_{j\in A_i} \min\{\hat{x}_j,\h z_i\}\ge u_i$. If so, set $\doublehat{z}_i=1$; else set $\doublehat{z}_i=0$.
\end{itemize}

Suppose the largest $z_i$ is $z_{i^*} = 1$, we have $\sum_{j\in A_{i^*}} y_{ij} \ge 1$. Therefore, for some $j$, $y_{i^*j}\ge 1/m$. Therefore $\sum_{j=1}^m {x_j}\ge \frac{1}{m}.$ Since Step 2 increases $\sum_j x_j$ by at most $\frac{1}{2m}$, we have  $\frac{\sum_j{x_j'}}{\sum_j{x_j}}\le 3/2.$ 

We first bound the expectation of $\hat{x}_j$. If $x_j'<1/2$, since $x_j'\ge \frac1{2m^2}$, we have:
\begin{align*}
    \mathbb{E}[\hat{x}_j]&=\int_{\alpha=0}^{2x_j'} 1 \, \d\alpha+\int_{\alpha=2x_j'}^1 2x_j'/\alpha\, \d \alpha =2x_j'+2x_j'\cdot \ln \alpha \Big\rvert_{2x_j'}^1 \le 2x_j'\cdot (1+2\ln m).  
\end{align*}
Therefore, we have 
\[
\mathbb{E}\Big[\sum_j \h x_j\Big]\le 2(1+2\ln m)\sum_j x_j'\le 3(1+2\ln m)\sum_j x_j.
\]

We now bound $\mathbb{E}\left[\sum_i \doublehat{z}_i\right]$. Let $P_i\triangleq\{j\in A_i: 2x_j'<\alpha\}$, 
$Q_i\triangleq\{j\in A_i:2x_j'\ge \alpha\}$. Since $\h x_j = 1$ for $j \in Q_i$, conditioned on $\h z_i = 1$, we have:
\begin{align*}
    \Pr\left(\doublehat{z}_i=0\right)&= \Pr\left(\sum_{j\in A_i} \min\{\h x_j,\h z_i\}< u_i+1 \right) =\Pr\left(\sum_{j\in P_i} \h x_j< u_i+1-|Q_i| \right).
\end{align*}

By the constraints in the optimization and since $y_{ij} = \min(x_j, z_i)$, we have 
$$\sum_{j\in P_i} \min\{x_j,z_i\}+\sum_{j\in Q_i}\min\{x_j,z_i\}\ge z_i\cdot(u_i+1).$$ 
Since the second term is capped by $z_i\cdot |Q_i|$, we have  $\sum_{j\in P_i} x_j\ge z_i\cdot \big((u_i+1)-|Q_i|\big).$
When $\h z_i=1$, we have $\alpha<z_i$, and thus 
$$\mathbb{E}\left[\sum_{j\in P_i}\h x_j\right]\ge 2\cdot\left(\sum_{j\in P_i} x_j'\right)\big/\alpha\ge 2\cdot \big((u_i+1)-|Q_i|\big)\cdot z_i/\alpha\ge 2\cdot \big((u_i+1)-|Q_i|\big).$$ 
By Chernoff Bounds on the independent binary random variables $\{\h x_j\}$, we have 
$$
\Pr\left(\sum_{j\in P_i} \h x_j< u_i+1-|Q_i|\,\Big\vert\, \h z
    _i=1\right)<\left(\frac{e^{-1/2}}{\sqrt{1/2}}\right)^2=2/e.
$$
Therefore, we have
\begin{align*}
\mathbb{E}\left[\sum_i \doublehat{z}_i\right]&= \sum_i \mathbb{E}\Big[\hat{z}_i\cdot \big(1-\Pr(\doublehat{z}_i=0\big)\Big] \ge \sum_i \mathbb{E}\Big[\hat{z}_i\cdot (1- \frac{2}{e})\Big] \ge (1- \frac{2}{e})\cdot \sum_i z_i. 
\end{align*}

Since $\{\h x_j\}$ and $\{\doublehat{z}_i\}$ form a valid solution to the program in \cref{sec:program}, there exists a setting of these variables such that
\[
\frac{\sum_j {\h x_j}}{\sum_i \doublehat{z}_i}\le \frac{\mathbb{E}[\sum_j {\h x_j}]}{\mathbb{E}[\sum_i \doublehat{z}_i]} \le  \frac{3(1+2\ln m)}{1-2/e}\cdot \frac{\sum_j x_j}{\sum_i z_i} = \frac{3(1+2\ln m)}{1-2/e}\cdot\theta_p.
\]
Therefore, we have $\theta_p \le \theta_c\le \frac{3(1+2\ln m)}{1-2/e}\cdot \theta_p$, completing the proof of the $O(\log m)$ approximation.

\paragraph{$O(\log n)$ Approximation.} We round the linear programming solution as follows. Recall that we assume  by scaling that the largest $z_i$ is $1$, and all $x_j, z_i \in [0,1]$.
\begin{enumerate}
    \item Construct $\omega=\lfloor\log_2 n\rfloor$ intervals: $I_0=[0, \frac{1}{n}], I_1=(\frac{1}{n}, \frac{2}{n}], \ldots, I_{\omega}=(\frac{2^{\omega}}{n}, 1]$. 
    \item If $z_i \in I_0$, set $\h z_i=0$. 
    \item Let $\ell^*=\arg\max_{\ell} \sum_{i:z_i\in I_\ell} z_i$. Let $I_{\ell^*}=(L^*, R^*]$. For all $z_i\notin I_{\ell^*}$, set $\h z_i=0$. 
    \item Set $x_j' =\min \{1, 2 x_j/L^*\}$. For each $j$ independently, set $\h x_j=1$ with probability $x_j'$ and $0$ otherwise.
    \item For each $i\in I_{\ell^*}$, if $\sum_{j\in A_i} \h x_j\ge (u_i+1)$, set $\h z_i = 1$, else set $\h z_i = 0$. 
\end{enumerate} 
We claim that 
$$\frac{\sum_{j=1}^m \hat{x}_j}{\sum_{i=1}^n \hat{z}_i}=O(\omega) \cdot \frac{\sum_{j=1}^m x_j}{\sum_{i=1}^n z_i}.$$ 
The second step will at most halve $\sum_{i=1}^n z_i$, since there are at most $n$ values $z_i$ in $I_0$ and the largest $z_i$ is normalized to 1. Further, the third step will at most multiply the objective function by $2\omega$. Let $Q = |\{i \mid z_i \in I_{\ell^*}\}|$. For the fourth step, we have:
$$ \frac{\sum_j x_j}{\sum_{i: z_i \in I_{\ell^*}} z_i} \ge \frac{\sum_j x_j}{Q\cdot  R^*} \ge \frac{\sum_j x_j}{Q\cdot  2L^*} \ge \frac{1}{4} \frac{\sum_j x'_j}{Q} = \frac{1}{4} \frac{\mathbb{E}[\sum_j \h x_j]}{Q}. $$
Let $S_1 = \{j \mid x_j \ge L^*/2\}$ and $S_2 = C \setminus S_1$. Note that for any $i \in I_{\ell^*}$, we have $x_j \le z_i$ for $j \in S_2$. For $i \in I_{\ell^*}$, the covering constraint in the LP relaxes to:
$$ \sum_{j \in S_2 \cap A_i} x_j \ge z_i\cdot  (u_i + 1 - |A_i \cap S_1|).$$
Dividing by $z_i/2$ (which is at least $L^*$) and observing that $x'_j \ge 2x_j/z_i$, we have:
$$ \sum_{j \in S_2 \cap A_i} x'_j \ge 2 (u_i + 1 - |A_i \cap S_1|). $$
We now apply Chernoff bounds to obtain:
$$ \Pr\left(\sum_{j\in S_2 \cap A_i} \h x_j< u_i + 1 - |A_i \cap S_1|\,\Big\vert\,  i \in I_{\ell^*}\right)<\left(\frac{e^{-1/2}}{\sqrt{1/2}}\right)^2=2/e.
$$
Therefore, $\mathbb{E}\left[\hat{z}_i\right]\ge (1-2/e)\cdot Q$. This shows an $O(\log n)$ approximation to $\theta_{c}$.



\subsection{Integrality Gap Instance}
 We now give an instance where there is an $\Omega(\min(\log m, \log n))$ integrality gap between $\theta_p$ and $\theta_c$, showing the analysis in \cref{sec:roundproof} is tight.

\begin{theorem}
There exists a committee s.t. $\theta_p =O\left(\frac{1}{\log \min(m, n)}\right)$ and $\theta_c=\Theta(1)$.
\end{theorem}
\begin{proof}
Suppose $n=2^{p+1}-2$ and there are $p$ groups of voters $V_1, \ldots, V_{p}$. For each $1\le \ell \le p$, $V_\ell$ contains $2^\ell$ voters. There are $p+2$ groups of candidates $C_1, \ldots, C_{p}, C',C^\mathrm{D}$. For each $q\in [p]$, $C_q$ is a group of $2^q-1$ candidates only approved by voter group $V_q$. $C'$ is a group of candidates with $2^{p}$ candidates approved by all voters. $C^\mathrm{D}$ is a group of $p$ ``dummy'' candidates approved by no one, added to ensure that $R = \frac{n}{k} = 1$. Consider the committee $W=C^\mathrm{D}\cup(\bigcup_{q=1}^{p} C_q)$ of size $n$. Note that $m = \Theta(n)$, and further, any voter in group $V_\ell$ has utility $2^\ell-1$ in $W$.

We first compute $\theta_c$. If a deviation has $z$ candidates, it can only make the utilities of voters in $V_1,V_2,\ldots, V_{\lfloor\log_2 z\rfloor}$ strictly increase. The total number of these voters is at most $2z-2$. Thus we have $\theta_c\ge \frac{z}{2z-2} \ge \frac{1}{2}$. 

We now compute the upper bound of $\theta_p$. Set $x_j = 2^{-p}$ for all $j \in C'$. For each $i \in V_{\ell}$, set $z_i = 1/2^{\ell}$. Note that $z_i \ge x_j$ for all $j \in C'$. We have $(U_i(W) + 1) \cdot z_i = 2^{\ell}\cdot 2^{-\ell} = 1$, while $\sum_{j \in C' \cap A_i} x_j = 1$. This is therefore a feasible LP solution, and its value is 
$\theta_p \le \frac{\sum_{j \in C'} x_j}{\sum_{i \in \bigcup_{\ell=1}^{p} V_{\ell}} z_i} = \frac{1}{p}.$
This shows a gap of $\Omega(p) = \Omega(\log m)$. The theorem follows since $n = \Theta(m)$.
\end{proof}

\section{Extension to Arbitrary Utilities and Sizes}
\label{sec:general}
We now extend the result in the previous section to the setting where the candidates have general sizes $s_j$, and voters have arbitrary additive utilities over candidates. We assume voter $i$ has utility $u_{ij} \in \mathbb{Z}^{+} \cup \{0\}$ for candidate $j$. Given a committee $W$ of size at most $k$, the utility of voter $i$ for the committee is $U_i(W) = \sum_{j \in W} u_{ij}$. We restrict the utilities to be integral, so that if $U_i(T) > U_i(W)$, then $U_i(T) \ge U_i(W) + 1$. Let $A_i = \{j \in C \mid u_{ij} > 0\}$.



\paragraph{LP Formulation.} A natural modification to the program in \cref{sec:program} for $\theta_c$ has unbounded integrality gap. 
We make two modifications to the linear program. First, in the optimal integer solution, we guess the candidate $j^*$ with largest size. This means we set $x_{j} = 0$ for all $j$ such that $s_j > s_{j^*}$, and delete these items. Since the numerator in the objective is at least $s_{j^*}$, we can set $x_j = 1$ for all $j$ with $s_j < \frac{s_{j^*}}{m}$, and this only increases the numerator by a constant factor. Let $S$ denote the set of these ``small'' items; we ignore these items, and set $U_i(W)$ to be $U_i(W) - U_i(S \cap A_i)$. If the latter quantity is smaller than zero, then we can set $z_i = 1$ and delete this voter from further consideration; this only lowers the objective. We let $m$ denote the number of candidates and $n$ denote the number of voters in the residual instance. We now scale the sizes so that the remaining items have sizes in $\left[\frac{1}{m},1\right]$. Let  $R = \frac{n}{k}$.

Next, we add knapsack cover constraints~\cite{Carr,KY}. Let $\h U_i(S) = \max(0,U_i(W)+1 - U_i(S))$, and let $u_{ijS} = \min(u_{ij}, \h U_i(S))$  The resulting LP is presented below. In this LP, first set of constraints can be interpreted as follows: Even if the $x_j$ for $j \in S$ are all set to $1$, so that voter $i$ already has utility $U_i(S)$, if voter $i$ is chosen by the integer program, the remaining $\{y_{ij}\}$ must push the total utility above $U_i(W)$. Further, any utility value $u_{ij}$ on the LHS can be truncated at $\h U_i(S)$ and the constraint should still hold. This constraint is clearly true for any $S$ in the integer program; the LP encodes the fractional version of all of them. 
\begin{gather*}
    \mbox{Minimize } R\cdot \frac{\sum_{j=1}^m s_j x_j}{\sum_{i=1}^n z_i},\ \text{s.t.}\\
    \forall i\in [n], S \subseteq [m],\ \sum_{j\in A_i \setminus S} u_{ijS} y_{ij}\ge z_i\cdot \h U_i(S);\\
    \forall i\in [n],\ \forall j\in A_i,\ y_{ij}\le \min(x_j, z_i);\\
    \forall i \in [n], j \in [m],\ x_j, z_i, y_{ij} \ge 0.
\end{gather*}

This LP has exponentially many constraints. For any given solution $(x,y,z)$ and fixed voter $i$, we divide the first set of constraints by $z_i$ and use the polynomial-time dynamic programming procedure exactly as in~\cite{Carr} to find the most violated constraint to a $(1+\epsilon)$ approximation, for constant $\epsilon > 0$. Omitting standard details, this implies the LP can be solved to a $(1+\epsilon)$ approximation in polynomial time via the Ellipsoid algorithm. 

\paragraph{Rounding.} 
The rounding is similar to \cref{sec:roundproof}, leading to the following theorem, whose proof is presented in \cref{app:gen_approx}. 

\begin{theorem}
\label{thm:arbitrary}
For the setting with arbitrary additive utilities and sizes, $\theta_c$ can be approximated to an $O(\min(\log m, \log n))$ factor in polynomial time.
\end{theorem}

\section{Auditing Lindahl Priceability}
\label{sec:comp}
In this section, we study fairness of a committee in terms of closeness to market clearing.  The concept is motivated by \emph{Lindahl equilibrium}~\cite{lindahl1958just,foley1970lindahl}, a market clearing concept for public goods. Such market clearing notions have been widely studied as fairness concepts in Economics~\cite{Fisher,varian}. Our main result is the following novel connection to the core --  auditing the approximation of a committee to Lindahl priceability reduces to the LP relaxation for auditing for core stability, hence leading to a polynomial time auditing algorithm. 

We consider the {\sc Approval Election} setting below. The extension to arbitrary utilities and sizes is presented in \cref{sec:lindahl_fractional}. 

\subsection{Lindahl Priceability}
As in the definition of core stability, we first scale the entitlements so that the entitlement of each voter is set to $1$ instead of $k/n$. Each candidate now requires $R = n/k$ entitlement to be paid for. A feasible committee of size $k$ corresponds to a total entitlement of $n$ in this scaling.

A committee $W$ of size at most $k$ is Lindahl priceable if there exists a price system $\{p_{ij}\}$ from voters to candidates, such that the following hold:
\begin{itemize}
    \item [1.] $\forall j\in [m],\ \sum_i p_{ij} \le R$, and
    \item [2.] $\forall i\in [n],\ T\subseteq C$, if $|T\cap A_i|\ge |W\cap A_i|+1$, then $\sum_{j \in T} p_{ij} > 1.$
\end{itemize}

The first condition above means that for each candidate, the prices from all voters sum up to at most $R = n/k$, so that each candidate is not ``over-paid''. Note that the first set of constraints can be made equalities by raising the prices $\{p_{ij}\}$, so the candidates are exactly paid for. The second condition means a voter cannot afford any committee that she strictly prefers to $W$. 

Lindahl priceability can be viewed as an integral version of the {\em gradient optimality conditions} in the fractional Lindahl equilibrium~\cite{foley1970lindahl}. As mentioned before, this makes our definition subtly different from a related concept in~\cite{Peters21b}. Analogous to the fractional Lindahl equilibrium, the following proposition holds, and we present a proof later in this section.  

\begin{proposition}
\label{prop1}
If a committee is Lindahl priceable, it lies in the core.
\end{proposition}

\subsection{Auditing via Duality}
As with core stability, we now define the best approximation to Lindahl priceability achievable by a committee $W$. Formally, we only allow a voter to use $\theta_p < 1$ endowment if they want to deviate to a committee with larger utility.

\begin{definition}[$\theta$-Approximate Lindahl Priceability]
\label{def:approxlindahl}
A committee $W$ of size at most $k$ is \emph{$\theta$-approximate Lindahl priceable} if there exists a price system $\{p_{ij}\}$ from voters to candidates, such that the following conditions hold:
\begin{itemize}
    \item [1.] $\forall j\in [m],\ \sum_i p_{ij} \le R$, and
    \item [2.] $\forall i\in [n],\ T\subseteq C$, if $|T\cap A_i|\ge |W\cap A_i|+1$, then $\sum_{j \in T} p_{ij} > \theta.$
\end{itemize}
\end{definition}

 The Lindahl priceability ratio of a committee $W$ is the smallest $\theta$ for which the committee is not $\theta$-approximate Lindahl priceable.  Our main result is the following theorem that ties Lindahl priceability ratio to the fractional relaxation of $\theta_c$. As a corollary, this shows that determining if a committee $W$ is Lindahl priceable is polynomial time solvable.  

\begin{theorem}
\label{thm:lindahl}
For a committee $W$, its Lindahl priceability ratio is  $\theta_p$ from \cref{eq:thetap}.
\end{theorem}
\begin{proof}
For simplicity, let $u_i = U_i(W)$. Let the Lindahl priceability ratio of the instance be $\theta_{\ell}$.  Fix the prices $\{p_{ij}\}$ achieving this. Then the minimum entitlement needed for a voter $i$ to deviate to a committee of utility larger than $u_i$ is captured by the following linear program:
\begin{gather*}
    \text{Minimize $ \sum_{j\in A_i} p_{ij} \gamma_{ij}$},\ \text{s.t.}\\
   \sum_{j\in A_i} \gamma_{ij}\ge u_i+1;\\
    \forall j\in A_i,\ \gamma_{ij}\le 1; \\
    \forall j\in A_i,\ \gamma_{ij}\ge 0.
\end{gather*}
Here, the variable $\gamma_{ij}$ corresponds to the fraction to which this voter chooses candidate $j$. In the optimal solution, these variables will be integers. Since the Lindahl priceability ratio is $\theta_{\ell}$, Condition (2) of \cref{def:approxlindahl} implies objective of the above LP is at least $\theta_{\ell}$ for any $i\in [n]$. 

Now take the dual of the above, where the dual variable for the first constraint is $\lambda_i$ and the dual variable for the second constraint is $\alpha_{ij}$. We obtain:
\begin{gather*}
\text{Maximize $\theta_i$},\ \text{s.t.}\\
    \forall j\in A_i, \lambda_i -  \alpha_{ij}\le p_{ij};\\
    (u_i+1)\lambda_i-\sum_{j\in A_i}\alpha_{ij} \ge \theta_i; \\
    \forall j \in [m],\  \lambda_i, \alpha_{ij} \ge 0.
\end{gather*}
Since the optimal $\theta_i \ge \theta_{\ell}$, this solution satisfies $(u_i+1)\lambda_i-\sum_{j\in A_i}\alpha_{ij} \ge \theta_{\ell}$. Since $\{p_{ij}\}$ satisfy Condition (1) in \cref{def:approxlindahl}, $\{p_{ij}\}, \{\alpha_{ij}\}, \{\lambda_i\}$ and $\theta = \theta_{\ell}$ are feasible for the following program:
\begin{gather*}
    \text{Maximize $\theta$},\ \text{s.t.}\\
    \forall i\in [n], j\in A_i, \lambda_i-  \alpha_{ij}\le p_{ij};\\
    \forall i\in [n],\ (u_i+1)\lambda_i-\sum_{j\in A_i}\alpha_{ij}\ge \theta;\\
    \forall j\in [m],\ \sum_{i\in T_j} p_{ij}\le R; \\
    \forall i \in [n], j \in [m], \ \lambda_i, \alpha_{ij}, p_{ij} \ge 0.
\end{gather*}
We now claim that the optimal solution to the above program must be exactly $\theta_{\ell}$.  If it is larger, this larger value $\theta'$ must be feasible for the per-voter duals, which means the per-voter primals have value at least $\theta'$. Then the Lindahl priceablility is at least $\theta'$, contradicting the definition of $\theta_{\ell}$. 

Finally, take the dual for the LP above, let $y_{ij},z_i,x_j$ respectively be the dual variable of the three constraints. The dual is the following:
\begin{gather*}
    \text{Minimize $R\cdot \sum_{j=1}^m{x_j}$},\ \text{s.t.}\\
    \forall i\in [n],\ \sum_{j\in A_i} y_{ij}\ge z_i\cdot (u_i+1);\\
    \forall i\in [n],\ \forall j\in A_i,\ y_{ij}\le x_j;\\
    \forall i\in [n],\ \forall j\in A_i,\ y_{ij}\le z_i;\\
     \sum_i z_i\ge 1; \\
    \forall i \in [n], j \in [m], \ z_i, x_j, y_{ij} \ge 0. 
\end{gather*}
This optimal value (which is $\theta_{\ell}$) is also the definition of $\theta_p$, completing the proof.
\end{proof}

Note that if $\theta_{\ell} > 1$, then since $\theta_c \ge \theta_p = \theta_{\ell} > 1$, we have $\theta_c > 1$. Therefore, if a committee is Lindahl priceable, it lies in the core, showing \cref{prop1}.

\section{ Sub-core for {\sc Approval Elections}}
\label{sec:subcore}
Given our approximation results for auditing the core, we can ask if such results can also be derived for weaker fairness notions. Such an auditing notion would be interesting to a practitioner in addition to auditing the core, since it is quite likely an implemented rule and resulting committee would be closer to satisfying a weaker but still reasonable notion of fairness compared to the core. We present a new weakening of the core for {\sc Approval elections}, that we term the sub-core, that we show also admits approximate auditing.  Note that this result is not implied by our results for the core; indeed, there are weakenings of the core, such as EJR, that we do not know how to efficiently audit.

\subsection{Weak Priceability}
In the multiwinner election setting, suppose the final condition in Lindahl priceability is relaxed so that each voter is only allowed to add candidates to its deviating committee, we get the following relaxed version of priceability. Recall that $R = n/k$, where $n$ is the total number of voters.

\begin{definition}[Weak Priceability]\label{def:subcore}
A committee $W$ of size at most $k$ is \emph{weakly priceable} if there exists a set of prices $\{p_{ij}\}$ from each voter $v_i$ to each candidate $c_j$, such that the following two conditions hold:
\begin{itemize}
    \item [1.] $\forall j\in [m],\ \sum_i p_{ij} \le R$.
    \item [2.] $\forall i\in [n]$, $d\in A_i\setminus W, p_{id}+\sum_{j \in A_i\cap W} p_{ij} > 1.$
\end{itemize}
\end{definition}
This notion is a relaxation of ``priceability'' as proposed in~\cite{peters2020proportionality} when $|W| = k$. The proof of the implication can be found in~\cite{berker2025edge}. 
Unlike Lindahl priceability, there are many natural and greedy voting rules, such as the Phragm\'en rule~\cite{phragmen}, that satisfy weak priceability, making it a desirable property to study in practice.

\subsection{Sub-core}
If we proceed as in the proof of \cref{thm:lindahl} and take the dual of the weak priceability ratio, we obtain a new concept of fairness that we call the {\em sub-core}. 

\begin{definition}[Sub-core]
A committee $W$ lies in the \emph{sub-core} if there is no $S\subseteq V$ and committee $T$ with $|T|\le \frac{|S|}{n}\cdot k$, s.t. $A_i\cap W \subsetneq A_i\cap T$ for all $i\in S$.
\end{definition}
The sub-core prevents any group of voters from deviating to a new committee in which each voter's approved candidates  forms a proper superset of the approved candidates in the original committee. 

Clearly, any committee that lies in the core also lies in the sub-core. The following proposition shows the sub-core is a weakening of weak priceability. 

\begin{proposition}\label{prop:subcore} 
If a committee is weakly priceable, then it lies in the sub-core.
\end{proposition}
\begin{proof}
For the purpose of contradiction, let $W$ be a committee that is weakly priceable but not in the sub-core. Suppose that voter set $S$ deviates to a committee $T$ such that $\forall i \in S, A_i\cap W\subsetneq A_i\cap T$. We consider all the prices for the candidates in $S$ and have the following inequalities, where the first inequality follows from the first constraint in \cref{def:subcore}:
\begin{align*}
    |T|\cdot R & \ge \sum_{j\in T} \sum_{i=1}^n p_{ij}= \sum_{i=1}^n \sum_{j\in T} p_{ij}\ge \sum_{i\in S}\left( \sum_{j\in T} p_{ij}\right)\\
    &\ge \sum_{i\in S} \left(\sum_{j\in A_i\cap W} p_{ij}+\sum_{j\in T\setminus W} p_{ij}\right)\\
    &>  \sum_{i\in S} 1 \tag*{(By the second constraint in \cref{def:subcore})}\\
    &=|S|.
\end{align*} 
This contradicts the fact that $|S|\ge |T|\cdot \frac{n}{k},$ completing the proof.  
\end{proof}

Since weakly priceable committees can be easily found by greedy procedures~\cite{PetersS20}, this shows that the sub-core is always non-empty.

\paragraph{Hardness of Auditing Sub-core.}
Though the sub-core seems like a weak and satisfiable fairness condition (it insists voters have no greedy deviation to a better committee), we show that deciding if a given committee $W$ lies in the sub-core is actually {\sc NP-Hard}.  
Towards this end, we observe that the core and sub-core coincide when each voter approves at most $2$ candidates (i.e., for all voters $i$, we have $|A_i| \le 2$). To see this, suppose the original committee was $W$ and a subset of voters deviate to $T$. If a deviating voter had original utility zero, then $A_i \cap W = \varnothing$, so that $A_i \cap T \supsetneq A_i \cap W$. Similarly, if $|A_i \cap W| = 1$ and $|A_i \cap T| = 2$, then $A_i \cap T = A_i \supsetneq A_i \cap W$. This shows any deviation satisfies the sub-core property, so that the core coincides with the sub-core. 

\begin{theorem}
If each voter only approves at most two candidates, deciding whether a committee $W$ does not lie in the sub-core (or core) is {\sc NP-Complete}. \label{thm:2_hard}
\end{theorem}

To prove \cref{thm:2_hard}, we reduce from the Maximal Independent Set problem for 3-regular graphs. 

\begin{lemma} [3-Regular Maximal Independent Set~\cite{fleischner2010maximum}]
Given a 3-regular graph $G(U,E)$, determine whether the Maximum Independent Set contains at least $q$ vertices is {\sc NP-hard}. 
\end{lemma}

\begin{proof}[Proof of \cref{thm:2_hard}]
For each instance $I=(q, G(U,E))$ where the degree of all vertices in $U$ is $D=3$, we add a vertex $s$ into $G$ and add an edge from each $v\in U$ to $s$. Suppose the new graph is $G'$. Suppose $|U|$ and $q$ are both sufficiently large. We consider the voting instance where the graph representation is $G'$. There are $|U|+1$ vertices, and each vertex represents a candidate. There are $|E|+|U|$ edges representing voters, and each voter approves 2 candidates corresponding to its two adjacent vertices. We let the committee size limit $k=\frac{(|E|+|U|)(q+1)}{q(D+1)}$.

We now prove that the committee $W=\{s\}$ lies in the sub-core if and only if there is no independent set in $G$ of size $q$.  First, if $W$ does not lie in the sub-core, then there is a deviation instance $(S,T)$ on $W$. Suppose that $T$ does not include $s$, then each candidate in $T$ is approved by at most $D$ voters in $T$, we have $|S|\le |T|\cdot D$. By the definition of deviation instance, we have $|T|\le |S|\cdot k/(|E|+|U|)=\frac{(q+1)|S|}{q(D+1)}\le |T|\cdot\frac{(q+1)D}{q(D+1)}<|T|$, leading to a contradiction. Therefore, we have $s\in T$. Since for each candidate $t\in T\setminus\{s\}$, there is an edge between $t$ and $s$, therefore there is a voter approving both $s$ and $t$. Other than these voters who approve $s$ and one candidate in $T_f\setminus\{s\}$, each candidate in $T$ is approved by at most $D$ voters.

Suppose that there are two candidate vertices $c_{j_1}, c_{j_2}$ in $T\setminus\{s\}$ sharing an edge $v$. We delete all edges which is adjacent to $c_{j_2}$ except $v$ from $S$ and delete $c_{j_2}$ from $T$. We denote the instance after the deletion by $(S',T')$, we originally have $|T|\le |S|\cdot \frac{(q+1)}{q(D+1)}$, the deletion decrease the left hand side by 1 and the right hand side by at most $D\cdot\frac{q+1}{q(D+1)}<1$. $(S', T')$ is still a deviation instance on $W$. We continue deleting such vertices and edges until there are no such pair of vertices in $T\setminus\{s\}$ sharing the same edge. We denote the final deviation instance by $(S_f,T_f)$.
    
Since the number of voters in $S_f$ is no more than the number of all edges adjacent to $T_f$, we have $|S_f|= (|T_f|-1)\cdot (D+1)$. If $|T_f|\le q$ then $|S_f|\cdot \frac{(q+1)}{q(D+1)}\le (|T_f|-1)\cdot \frac{q+1}{q}<|T_f|$. This contradicts $|T|\le |S|\cdot \frac{q+1}{q(D+1)}$. Therefore we have $|T_f|\ge q+1$. Since $T_f\setminus\{s\}$ is an independent set on $G$ and $|T_f\setminus\{s\}|\ge q$, there exists an independent set of size $q$ in the graph $G$.  
    
Conversely, if there is an independent set in $G$ of size $q$, then we select the corresponding candidate vertices in the independent set from $G'$ and $s$ to form $T$. We let $S$ be all voters whose corresponding edges are adjacent to vertices $T$. Since the $T\setminus \{s\}$ does not contain two vertices share the same edge, we have $|S|=(|T|-1)\cdot (D+1)= q\cdot (D+1)$. Since $|S|\cdot \frac{q+1}{q(D+1)}=q+1=|T|$, $(S,T)$ is a deviation instance on $W$ and thus $W$ does not lie in the sub-core.   

Therefore, $G$ has an independent set of size $q$ if and only if the corresponding committee $W$ does not lie in the sub-core and this completes the reduction. 
\end{proof}

\paragraph{Approximately Auditing Sub-core Property.} Similar to $\theta_c$, we can now define a parameter $\theta_{sc}$ showing how close a committee is to the sub-core. 

\begin{definition}
For $\theta \le 1$, a committee $W$ of size $k$ lies in the \emph{$\theta$-approximate sub-core} if for all subsets of voters $T \subseteq [n]$, there is no {\em deviating committee} $T'$ with size at most $\theta \cdot |T| \cdot \frac{k}{n}$, such that for all $i \in T$, we have $A_i\cap W \subsetneq A_i\cap T'$.
\end{definition}

Given a committee $W$, $\theta_{sc}$ is defined as the smallest $\theta$ such that $W$ is not in the $\theta$-approximate sub-core.  We prove the following theorem in \cref{app:subcore}, that shows the sub-core can be approximately audited. This positive result on auditing makes the sub-core a desirable weakening of the core property. The proof is largely similar to the proof of \cref{thm:log_approx}. 

\begin{theorem} [Proved in \cref{app:subcore}]
\label{thm:subcore}
Given any committee $W$, $\theta_{sc}$ has an $O(\min(\log m, \log n))$ approximation in polynomial time, where $m,n$ are the total number of candidates and voters respectively.
\end{theorem}

\section{Conclusion}
Note that our theoretical approximation results for auditing are worst case guarantees. In practice, the linear program value $\theta_p$ will provide a lower bound on $\theta_c$, and if this can be rounded so that the integer solution has value $\alpha \theta_p$ for some small $\alpha \ge 1$, then this sandwiches $\theta_c \in [\theta_p, \alpha \theta_p]$. Further, the rounding outputs a deviating set of voters and their chosen committee, which will be of interest as a demographic that is not well-represented by the current committee.

The main open question arising from this work is closing the gap between the positive and hardness results for auditing the core. As mentioned before, showing such results for density objectives is challenging~\cite{Khot-densest}. A related question is existence results: A major open question in social choice is whether there is a committee in the core for {\sc Approval Elections}. Though there are voting rules that find committees in the approximate core~\cite{peters2020proportionality,DBLP:conf/stoc/JiangMW20,ChenFLM19}, these results do not translate to the exact core. Even more specifically, it is an open question whether there is always a committee that is Lindahl priceable. 

Finally, it would be interesting to use the techniques in this paper to approximately audit other notions of fairness or efficiency in social choice. For instance, consider the notion of {\em extended justified representation} (or EJR,~\cite{Aziz16}), where a group of $t \cdot n/k$ voters can only deviate if they all approve at least $t$ candidates in common. Since this notion is weaker than the core, it is easier to show existence -- indeed the PAV rule~\cite{thiele1895om} satisfies EJR but fails the core. However, imposing restrictions on the deviation does not necessarily make it easier to audit such notions~\cite{Peters21b}, and we do not know how to audit EJR approximately. 
We showed a particular weakening of the core, the sub-core, that can be approximately audited, and it would  be interesting to study the landscape of efficient auditing more systematically.

\paragraph{Acknowledgment. } The original version of the paper erroneously claimed that Weak Priceability in \cref{sec:subcore} was equivalent to the priceability definition in~\cite{peters2020proportionality}. We thank Emin Berker for pointing out that our definition is strictly weaker.

\bibliographystyle{abbrv}
\bibliography{ref}

\appendix
\section{Other Hardness Results from \cref{sec:hard}}
\label{app:hard}
Recall that we are considering the {\sc Approval Election} setting. We present some additional hardness results for computing $\theta_c$.

\subsection{Proof of \cref{thm:hard1}}
\label{app:hard1}
The following result is shown in~\cite{Corneil}:

\begin{lemma}[Corollary 2.4 in~\cite{Corneil}]
Given a Hamiltonian undirected graph $G(V,E)$ with degree at most $6$, determining if it has a $3$-regular subgraph is {\sc NP-Complete}.
\end{lemma}

Given such a Hamiltonian graph $G(V,E)$ with $n$ vertices and $m \le 3n$ edges, we make each edge a candidate and each vertex a voter, who approves her incident edges. Therefore, each voter approves at most $6$ candidates and each candidate is approved by at most $2$ voters. The original committee $W$ that needs to be audited is the Hamiltonian subgraph. This gives each voter in $V$ a utility of exactly $2$, so that they need a utility of at least $3$ to deviate. We set $k = 3n/2$, so that $R = n/k = \frac{2}{3}$. 

Suppose $G$ has a $3$-regular subgraph $H(V',E')$. Consider the set of voters $V'$ deviating to committee $E'$. Since $|E'| = 1.5 |V'|$, we have:
$$ \theta_c \le R \cdot\frac{|E'|}{|V'|} = 1. $$

On the other hand, suppose $G$ does not have a $3$-regular subgraph. Let $V'$ be any set of voters that deviate using a set $E'$ of edges. Each voter has at least $3$ edges incident. This means either some $v \in V'$ has degree at least $4$, or there is some $e \in E'$ such that one of its end-points is not in $E'$. In either case, $2|E'| > 3 |V'|$, so that $\theta_c > 1$. Therefore, deciding if $\theta_c \le 1$ (that is, if $W$ does not lie in the core) on such instances is {\sc NP-Hard}.

\subsection{Hardness of Auditing Fractional Committees}
\label{app:hard2}
We will next show that the auditing problem remains hard to approximate even when candidates can be chosen {\em fractionally}. This setting models several real-world participatory budgeting elections, where the amount of money allocated to a project can be continuous within a range. Mathematically, this corresponds to allowing the variables $\{x_j\}$ to be fractional in $[0,1]$ in the program in \cref{sec:program}, while insisting voters deviate only if they get a fixed amount of additional utility on deviation. Formally, 

\begin{definition}
\label{def:fractional_core}
For $\theta \le 1$ and constant $\eta \in (0,1]$, a committee $\vec{x} \in [0,1]^m$ with $\sum_j s_j x_j \le k$ lies in the \emph{$(\theta,\eta)$-approximate fractional core} if for all $S \subseteq [n]$, there is no \emph{deviating committee} $\vec{y} \in [0,1]^m$ with  $ \sum_j s_j y_j \le \theta\cdot |S|\cdot \frac{k}{n}$, such that for all $i \in S$, we have $U_i(\vec{y}) \ge U_i(\vec{x}) + \eta$.
\end{definition}


We show the following hardness of approximation result.

\begin{theorem}\label{thm:np_hard_frac}
For any $0<\eta\le 1$ and any $\gamma > 0$, distinguishing instances that do not lie in the $(\theta_c, \eta)$-approximate fractional core from those that lie in the $(\theta_c (1.1839 - \gamma), \eta)$-approximate fractional core  is {\sc NP-Hard}.
\end{theorem}
\begin{proof}
The proof is a reduction from the Regular Maximum Coverage(\cref{lem:regular_cover}) and is similar to that in \cref{thm:standard_inapprox}. Based on \cref{lem:regular_cover}, we construct the following instance for auditing the core:
\begin{itemize}
    \item There are $\xi$ main candidates. Each candidate corresponds to a set in the covering instance. 
    \item There are 3 groups of voters, denoted by $G_1$, $G_2$ and $G_3$. 
    \item $G_1$ contains $qd$ voters. Each of these voters corresponds to an element of the covering instance and approves the main candidates whose corresponding set contains their corresponding element. 
    \item $G_2$ contains $\xi\cdot pd$ voters. These are divided into $\xi$ subgroups of $pd$ voters and each subgroup approves a single main candidate. 
    \item For each voter in $G_1$ and $G_2$, we add one dummy candidate $j$ that only this voter approves. In the allocation $\vec{y}$ that we seek to audit, we set $y_j = 1-\eta$. Therefore, each voter in $G_1$ and $G_2$ has initial utility $1-\eta$. 
    \item   $G_3$ contains $c\cdot qd$ voters. For each voter in $G_3$, we add $\lceil q-\eta\rceil$ dummy candidates that only this voter approves. For each such candidate $j$, we set $y_j = \frac{q-\eta}{\lceil q-\eta\rceil}$. Thus, the initial utility of each voter in $G_3$ is $q-\eta$. These voters also approve all the main candidates.
\end{itemize}
In total, there are $m_0 = \xi\cdot pd+c\cdot qd\cdot \lceil q-\eta\rceil$ dummy candidates, each approved by one voter. There are in total $n = qd+\xi\cdot pd+c\cdot qd$ voters. We choose $k \ge m_0$ and set $R = n/k$.

We say a voter is ``satisfied'' if their utility increases by at least $\eta$ on deviation. In order to make each voter in $G_1$ and $G_2$ satisfied, their utility must be at least $1$. In order to make each voter in $G_3$ satisfied, their utility must be at least $q$. 

If the covering instance is a ``YES'' instance, we can select $q$ main candidates that cover the $qd$ voters in $G_1$. Each main candidate can also satisfy $pd$ voters in $G_2$. Since we have chosen $q$ candidates, all voters in $G_3$ are also satisfied. Therefore, these $q$ candidates can satisfy at least $qd+ qpd+cqd$ voters. We have $$\theta_c\le R\cdot \frac{q}{qd+p\cdot qd+c\cdot qd}=\frac{R}{d}\cdot \frac{1}{1+p+c}.$$

If the covering instance is a ``NO'' instance, we first argue that there is a deviating (optimal) allocation $\vec{x}$ where no dummy candidate has strictly positive allocation. Suppose that that $x > 0$ is the allocation to a dummy candidate $c_d$ and $c_d$ is only approved by the voter $v_d$. If $v_d$ is not satisfied, then we can reduce $x$ to $0$ to decrease $\theta_c$. If $v_d$ is satisfied, consider the main candidates she approved without $c_d$. Since all voters can be satisfied had the main candidates all been chosen integrally, we can reduce $x_j$ to $0$ and increase the main candidates' allocation by at most $x_j$ to make $v_d$ satisfied. This does not increase $\theta_c$. Therefore, without loss of generality, we can assume that no dummy candidate has positive allocation in the deviation. 

For the main candidates, suppose that the sum of integral allocations ($x_j=1$) is $S_1$ and the sum of fractional allocations ($0<x_j<1$) is $S_2$. 

If $S_1+S_2<q$, then the voters in $G_3$ cannot be satisfied. Since all voters in $G_1$ need utility at least $1$ to be satisfied, at most $d\cdot (|S_1|+|S_2|)$ of voters in $G_1$ are satisfied. Since each integrally chosen main candidate can satisfy $pd$ voters in $G_2$, we can satisfy at most $d\cdot(|S_1|+|S_2|)+pd\cdot |S_1|$ voters in $G_2$. Therefore, we have $$
\theta_c\ge R \cdot\frac{|S_1|+|S_2|}{|S_1|\cdot (1+p)d+|S_2|\cdot d}\ge \frac{R}{d}\cdot\frac{1}{1+p}.
$$

If $|S_1|+|S_2|\ge q$, then all the voters in $G_3$ are satisfied. Suppose that $|S_1|\le (1-1/e)\cdot q$, then we have $S_2\ge q/e$. Let $\kappa =|S_2|/(|S_1|+|S_2|)\ge 1/e$. We can satisfy at most $(|S_1|+|S_2|)\cdot d$ voters in $G_1$. Since a main candidate needs to be integrally open to satisfy a voter in $G_2$, we can satisfy at most  $|S_1|\cdot p\cdot d$ voters in $G_2$.  Therefore, we have
$$
\theta_c \ge R \cdot\frac{|S_1|+|S_2|}{(|S_1|+|S_2|)\cdot d+|S_1|\cdot p\cdot d+c\cdot qd}\ge \frac{R}{d}\cdot \frac{1}{1+(1-\kappa)\cdot p+c}\ge \frac{R}{d}\cdot \frac{1}{1+(1-1/e)\cdot p+c}.
$$

If $(1-1/e)\cdot q\le |S_1|\le  q$, then using only $S_1$, we can satisfy at most $(1-1/e+\eps)\cdot q\cdot d$ voters in $G_1$. After adding fractional allocations $S_2$, we can satisfy at most $|S_2|\cdot d$ more voters on $G_1$. We can further satisfy at most $|S_1|\cdot p\cdot d$ voters in $G_2$, and all voters in $G_3$. Using $\kappa =|S_2|/(|S_1|+|S_2|)$,  we have
\begin{align*}
    \theta_c &\ge R \cdot\frac{|S_1|+|S_2|}{(1-1/e+\eps)\cdot q\cdot d+|S_1|\cdot p\cdot d+c\cdot q\cdot d+|S_2|\cdot d}\\
    &\ge \frac{R}{d} \cdot\frac{1}{(1-1/e+\eps)+(1-\kappa)\cdot p+c+ \kappa}.
\end{align*}

If $|S_1|>q$, then we first select any subgroup of $q$ candidates from $S_1$. These candidates can satisfy at most $(1-1/e+\eps)\cdot dq$ voters in $G_1$. After adding the rest of $S_1$ and $S_2$ we can satisfy at most $(|S_2|+|S_1|-q)\cdot d$ more voters in $G_1$. We can further satisfy at most $|S_1|\cdot pd$ voters in $G_2$ and all voters in $G_3$. Denote $r=\frac{|S_1|+|S_2|}q -1$. Therefore, we have
\begin{align*}
    \theta_c &\ge R \cdot\frac{|S_1|+|S_2|}{(1-1/e+\eps)\cdot q\cdot d+|S_1|\cdot p\cdot d+c\cdot qd+(|S_1|+|S_2|-q)\cdot d}\\
    &\ge R \cdot\frac{|S_1|+|S_2|}{(1-1/e+\eps)\cdot q\cdot d+(|S_1|+|S_2|)\cdot p\cdot d+c\cdot qd+(|S_2|+|S_1|-q)\cdot d}\\
    &\ge \frac{R}{d} \cdot\frac{1+r}{(1-1/e+\eps)+(1+r)\cdot p+c+r}\\
    &= \frac{R}{d} \cdot\frac{1}{\frac{-1/e+c+\eps}{1+r}+ 1 +p}.
\end{align*}

By setting $p=1$, $c=\kappa=1/e$ and taking the minimum over the previous four lower bounds, we have $\theta_c\ge \frac{R}{d}\cdot \frac{1}{2+\eps}$ if the covering instance is ``NO'' instance. If the covering instance is ``YES'', we have $\theta_c\le \frac{R}{d}\cdot \frac{1}{2+1/e}$. Since the gap of $\theta_c$ between the constructed instance from ``YES'' and from ``NO'' instances is at least $\frac{2+1/e}{2+\eps}=1+\frac{1}{2e}-o(1)$, approximating the $\theta_c$ to within this factor is {\sc NP-Hard}.
\end{proof}

\section{General Utilities and Sizes: Proof of \cref{thm:arbitrary}}
\label{app:gen_approx}
\subsection{$O(\log m)$ Approximation}
The rounding algorithm is similar to that in \cref{sec:roundproof}. We first scale the solution so that all variables are in $[0,1]$, and $\max_i z_i = 1$. We now apply the following steps:
\begin{itemize}
    \item [1.] Pick $\alpha \in [0,1]$ uniformly at random.
    \item [2.] Let $x_j'=\max\{\frac{1}{2m^2},x_j\}$.
    \item [3.] If $z_i> \alpha$, then set $\hat{z}_i=1$; else set $\hat{z}_i=0$.
    \item [4.] If $2x_j'> \alpha$, then set $\hat{x}_j=1$; else set $\hat x_j=1$ with probability $2x_j'/\alpha$. We round each $\hat{x}_j$ independently.
\end{itemize}

At this point, for each $i$ with $\hat z_i=1$, let $P_i\triangleq\{j\in A_i: 2x_j'<\alpha\}$ and $Q_i\triangleq\{j\in A_i:2x_j'\ge \alpha\}$. Note that $\h x_j = 1$ for $j \in Q_i$. Let $\h u_i = \max(0,U_i(W) + 1 - U_i(Q_i))$, and $\h u_{ij} = \min(u_{ij}, \h u_i)$. Since we assumed $u_{ij}$ are integers, these values remain non-negative integers. The LP variables satisfy the knapsack cover constraint for set $P_i$:
\begin{equation}
\label{eq:knap_cover2}
    \sum_{j \in P_i} \h u_{ij} \min(x_j, z_i) \ge  z_i \h u_i.
\end{equation}

For $i$ with $\h z_i = 1$, as a final step, check if $\sum_{j \in P_i} \h u_{ij}  \hat{x}_j \ge \h u_i$. If so, set $\doublehat{z}_i=1$; else set $\doublehat{z}_i=0$.

\paragraph{Analysis Sketch.} The analysis follows the same outline as \cref{sec:roundproof}, and we highlight the differences. Consider the $i$ for which $z_i = 1$. The knapsack cover inequalities for this voter now imply
$$\sum_j \min(U_i(W) + 1, u_{ij}) x_j \ge U_i(W) + 1 \qquad \Rightarrow \qquad \sum_j x_j \ge 1.$$ 
Since all $s_j \in [1/m,1]$, we have $\sum_j s_j x_j \ge 1/m$. Since $\sum_j s_j (x'_j - x_j) \le \sum_j \frac{1}{2m^2} \le \frac{1}{2m}$, the increase in objective by using $x'_j$ instead of $x_j$ remains a constant factor. Proceeding as before, 
$$\mathbb{E}\left[\sum_j s_j \h x_j\right] = O(\log m) \sum_j s_j x_j.$$

We now bound $\mathbb{E}\left[\sum_i \doublehat{z}_i\right]$. Since $\h x_j = 1$ for $j \in Q_i$, we have
\begin{align*}
    \Pr\left(\doublehat{z}_i=0 \mid \h z_i=1\right)&= \Pr\left(\sum_{j\in P_i}  \h u_{ij}  \hat{x}_j < \h u_i\,\Big\vert\, \h z_i=1\right).
\end{align*}

When $\h z_i=1$, we have $\alpha<z_i$. Using \cref{eq:knap_cover2}, we have:
$$\mathbb{E}\left[\sum_{j\in P_i}  \h u_{ij} \h x_j\right]
\ge 2\cdot\left(\sum_{j\in P_i}  \h u_{ij} x_j'\right)\big/\alpha 
\ge 2\cdot \h u_i \cdot z_i/\alpha\ge 2\cdot \h u_i.$$ 
Note that $\h u_{ij} \le \h u_i$ which allows us to use Chernoff Bounds. We therefore have:
$$
\Pr\left(\sum_{j\in P_i} \h u_{ij} \h x_j  < \h u_i \right)< e^{-1/4}.
$$
By linearity of expectation,
$$ \mathbb{E}\left[\sum_i \doublehat{z}_i\right] = \sum_i \mathbb{E}\Big[\hat{z}_i\cdot \big(1-\Pr(\doublehat{z}_i=0 \mid \h z_i=1)\big)\Big] \ge \sum_i \mathbb{E}\Big[\hat{z}_i\cdot (1-e^{-1/4})\Big] \ge \frac{1}{5} \cdot \sum_i z_i. 
$$
Note that if $\sum_{j\in P_i} \h u_{ij} \h x_j  \ge \h u_i$, then $\sum_{j \in P_i \cup Q_i} u_{ij} \h x_j \ge U_i(W) + 1$, so that the constraint for voter $i$ is indeed satisfied if $\doublehat{z}_i = 1$. Taking the ratio of $\sum_j \h x_j$ and $\sum_i \doublehat{z}_i$ just as before, this yields an $O(\log m)$ approximation to $\theta_c$.  

\subsection{$O(\log n)$ Approximation}
Again, the rounding algorithm is similar to that in \cref{sec:roundproof}. We first scale the solution first so that all variables are in $[0,1]$, and $\max_i z_i = 1$. We now apply the following steps:
\begin{enumerate}
    \item Construct $\omega=\lfloor\log_2 n\rfloor$ intervals: $I_0=[0, 1/n], I_1=(1/n, 2/n], \ldots, I_{\omega}=(2^{\omega}/n, 1]$.
    \item If $z_i \in I_0$, set $\h z_i=0$. 
    \item Let $\ell^*=\arg\max_{\ell} \sum_{i:z_i\in I_\ell} z_i$. Suppose $I_{\ell^*}=(L^*, R^*]$. For all $z_i\notin I_{\ell^*}$, set $\h z_i=0$. 
    \item Set $x_j' =\min \{1, 2 x_j/L^*\}$. For each $j$ independently,  $\h x_j=1$ w.p. $x_j'$ and $0$ otherwise.
\end{enumerate}

Let $S_1 = \{j \mid x_j \ge L^*/2\}$ and $S_2 = C \setminus S_1$. Note that for any $i \in I_{\ell^*}$, we have $x_j \le z_i$ for $j \in S_2$. Further, note that $x_j' = 1$ for $j \in S_1$. Let $\h u_i = \max(0,U_i(W) + 1 - U_i(S_1))$, and $\h u_{ij} = \min(u_{ij}, \h u_i)$. Since we assumed $u_{ij}$ are integers, these values remain non-negative integers. The LP variables satisfy the knapsack cover constraint for set $S_2$:
\begin{equation}
\label{eq:knap_cover3}
    \sum_{j \in S_2} \h u_{ij} \min(x_j, z_i) \ge  z_i \h u_i \qquad \forall i \in I_{\ell^*}.
\end{equation}

For each $i\in I_{\ell^*}$, check if $\sum_{j\in S_2} \h u_{ij} \h x_j \ge \h u_i$. If so, set $\h z_i = 1$, else set $\h z_i = 0$. 

\paragraph{Analysis Sketch.} The analysis follows the same outline as \cref{sec:roundproof}, and we highlight the differences. As in that proof, letting $Q = |\{i \mid z_i \in I_{\ell^*}\}|$, we obtain:
$$ \omega \cdot \frac{\sum_j s_j x_j}{\sum_i z_i} \ge  \frac{\sum_j s_j x_j}{\sum_{i: z_i \in I_{\ell^*}} z_i} \ge \frac{\sum_j s_j x_j}{Q\cdot  R^*} \ge \frac{\sum_j s_j x_j}{Q\cdot  2L^*} \ge \frac{1}{4} \frac{\sum_j s_j x'_j}{Q} = \frac{1}{4} \frac{\mathbb{E}[\sum_j s_j \h x_j]}{Q}. $$
Dividing \cref{eq:knap_cover3} by $z_i/2$ (which is at least $L^*$) and observing that $x'_j \ge 2x_j/z_i$, we have:
$$  \sum_{j \in S_2} \h u_{ij} x_j' \ge  2 \h u_i \qquad \forall i \in I_{\ell^*}. $$
Since $\h u_{ij} \le \h u_i$, we now use Chernoff bounds to obtain for any $i \in I_{\ell^*}$:
$$ \Pr \left(\sum_{j\in S_2} \h u_{ij} \h x_j  < \h u_i \right)< e^{-1/4}.$$
Therefore, $\mathbb{E}\left[\sum_i \hat{z}_i\right]\ge \frac{Q}{5}$. Therefore,
$$ \frac{\mathbb{E}[\sum_j s_j \h x_j]}{\mathbb{E}[\sum_i \h z_i]} \le 20 \omega  \cdot \frac{\sum_j s_j x_j}{\sum_i z_i},$$
which implies an $O(\log n)$ approximation to $\theta_{c}$.

\section{Lindahl Pricability for General Utilities and Sizes}
\label{sec:lindahl_fractional}
We consider the  Participatory Budgeting setting where the utilities and sizes are general. Let $u_{ij}$ denote the utility of candidate $j$ for voter $i$, and let $s_j$ denote the size of candidate $j$. 

\subsection{Fractional Allocations}
We first suppose candidates can be allocated fractionally, so that the allocation is represented as $x_j \in [0,1]^n$. Denote the utility of voter $i$ as $U_i(\vec{x}) = \sum_{j \in C} u_{ij} x_j$. We define approximate Lindahl priceability as follows:

\begin{definition}[$(\theta,\eta)$-Approximate Fractional Lindahl Priceability]
\label{def:approxlindahl2}
For constant $\eta > 0$, a committee $\vec{x} \in [0,1]^m$ with $\sum_j s_j x_j \le k$ is \emph{$(\theta,\eta)$-approximate Lindahl priceable} if there exists a price system $\{p_{ij}\}$ from each voter $v_i$ to each candidate $c_j$, such that the following two conditions hold:
\begin{itemize}
    \item [1.] $\forall j\in [m],\ \sum_i p_{ij} \le R$.
    \item [2.] $\forall i\in [n],\ \vec{y} \in [0,1]^m$, if $U_i(\vec{y}) \ge U_i(\vec{x}) + \eta$, then $\sum_{j \in C} p_{ij} s_j y_j > \theta.$
\end{itemize}
\end{definition}

Given a committee $\vec{x}$ and a constant $\eta > 0$, let $\theta_{\ell}$ denote the smallest $\theta$ for which the committee is not $(\theta, \eta)$-Lindahl priceable. We can extend the result in \cref{sec:comp} to show the following theorem:

\begin{theorem}
The value of $\theta_{\ell}$ can be computed in polynomial time.
\end{theorem}
\begin{proof}
Fix a committee $\vec{x}$. Following the proof in \cref{thm:lindahl}, let $u_i = U_i(\vec{x}) = \sum_j u_{ij} x_j$.  Fix the prices $p_{ij}$ achieving the Lindahl priceability ratio. Then the minimum entitlement needed for a voter $i$ to deviate to a committee of utility at least $u_i + \eta$ is captured by the following linear program, where $\gamma_{ij} \in [0,1]$ is the fraction to which candidate $j$ is allocated in the solution.
\begin{gather*}
    \text{Minimize $ \sum_{j\in A_i} p_{ij} s_j \gamma_{ij}$},\ \text{s.t.}\\
   \sum_{j \in C} u_{ij} \gamma_{ij}\ge u_i+ \eta;\\
    \forall j\in C,\ \gamma_{ij}\le 1; \\
    \forall j\in C,\ \gamma_{ij}\ge 0.
\end{gather*}
We take the dual as in that proof, and finally obtain the following program:
\begin{gather*}
    \text{Maximize $\theta$},\ \text{s.t.}\\
    \forall i\in [n], j\in [m], u_{ij} \lambda_i- \alpha_{ij}\le p_{ij} s_j;\\
    \forall i\in [n],\ (u_i+ \eta)\lambda_i-\sum_{j\in C}\alpha_{ij}\ge \theta;\\
    \forall j\in [m],\ \sum_{i} p_{ij}\le R; \\
    \forall i \in [n], j \in [m], \ \lambda_i, \alpha_{ij}, p_{ij} \ge 0.
\end{gather*}
The optimal $\theta$ corresponds to $\theta_{\ell}$, completing the proof.
\end{proof}

\subsection{Integer Allocations}
We next consider the case where candidates need to be allocated integrally. For committee $W \subseteq C$, let $U_i(W) = \sum_{j \in W} u_{ij}$. We assume utilities are integers. We define approximate Lindahl priceability as follows:

\begin{definition}[$\theta$-Approximate Lindahl Priceability]
\label{def:approxlindahl3}
A committee $W$ with $\sum_{j \in W} s_j \le k$ is \emph{$\theta$-approximate Lindahl priceable} if there exists a price system $\{p_{ij}\}$ from each voter $v_i$ to each candidate $c_j$, such that the following two conditions hold:
\begin{itemize}
    \item [1.] $\forall j\in [m],\ \sum_i p_{ij} \le R$.
    \item [2.] $\forall i\in [n],\ T\subseteq C$, if $U_i(T) \ge U_i(W) +1$, then $\sum_{j \in T} p_{ij} s_j > \theta.$
\end{itemize}
\end{definition}

Given a committee $W$, let $\theta_{\ell}$ be the smallest $\theta$ for which the committee is not $\theta$-approximately Lindahl priceable. We show the following theorem:

\begin{theorem}
For any constant $\epsilon > 0$, a $(2 + \epsilon)$-approximation to $\theta_{\ell}$ can be computed in polynomial time.
\end{theorem}
\begin{proof}
Fix a committee $W$. Following the proof in \cref{thm:lindahl}, let $u_i = U_i(W) = \sum_{j \in W} u_{ij}$.  Fix the prices $p_{ij}$ achieving the Lindahl priceability ratio. Then the minimum entitlement needed for a voter $i$ to deviate to a committee of utility at least $u_i + 1$ is captured by the following knapsack cover linear program, where for any $S \subseteq C$, we define $\h U_i(S) = \max(0,u_i +1 - U_i(S))$;  $u_{ijS} = \min(u_{ij}, \h U_i(S))$; and $A_i = \{j \in C \mid u_{ij} > 0\}$:
\begin{gather*}
    \text{Minimize $ \sum_{j\in A_i} p_{ij} s_j \gamma_{ij}$},\ \text{s.t.}\\
    \forall i\in [n], S \subseteq A_i,\ \sum_{j\in A_i \setminus S} u_{ijS} \gamma_{ij}\ge \h U_i(S); \\
    \forall j\in C,\ \gamma_{ij}\le 1; \\
    \forall j\in C,\ \gamma_{ij}\ge 0.
\end{gather*}
Though this LP may allocate $\gamma_{ij}$ fractionally, it is known~\cite{Carr} that the optimum objective is within a factor of $2$ of the integer optimum. As in the proof of \cref{thm:lindahl}, now take the dual, put these duals together for all voters adding the constraint $\sum_i p_{ij} \le R$ for all voters $i \in [n]$, and take the dual again. This yields the following LP, which is identical to that in \cref{sec:general}.
\begin{gather*}
    \mbox{Minimize } R\cdot \sum_{j=1}^m s_j x_j,\ \text{s.t.}\\
    \forall i\in [n], S \subseteq [m],\ \sum_{j\in A_i \setminus S} u_{ijS} y_{ij}\ge z_i\cdot \h U_i(S);\\
    \forall i\in [n],\ \forall j\in A_i,\ y_{ij}\le x_j;\\
    \forall i\in [n],\ \forall j\in A_i,\ y_{ij}\le z_i; \\
    \sum_{i=1}^n z_i \ge 1; \\
    \forall i \in [n], j \in [m],\ x_j, z_i, y_{ij} \ge 0.
\end{gather*}
This yields a $2$-approximation to $\theta_{\ell}$. The LP above has polynomially many variables but exponentially many constraints. It is shown in~\cite{Carr} that given a setting of variables, the most violated constraint can be computed to a $(1+\epsilon)$-approximation in polynomial time for any constant $\epsilon > 0$. This implies the LP can be solved in polynomial time to a $(1+\epsilon)$-approximation via the Ellipsoid method. This completes the proof of the theorem.
\end{proof}

\section{Approximately Auditing Sub-core: Proof of \cref{thm:subcore}}
\label{app:subcore}
We can compute $\theta_{sc}$ by the following Mathematical Program. 
\begin{gather*}
    \mbox{Minimize}\ \ \  R\cdot  \frac{\sum_{j=1}^m x_j}{\sum_{i=1}^n z_i}, \ \text{s.t.}\\
    \forall i\in [n], j\in W \cap A_i, x_j \ge z_i,\\  
    \forall i\in [n], \sum_{j\in A_i\setminus W} x_j\ge z_i,\\
    \forall i\in [n], z_i \in \{0, 1\},\\
    \forall j\in [m], x_j \in \{0, 1\}.
\end{gather*}
The first constraint means that if a voter deviates, all her approved candidates in the original committee must also be selected in the deviating committee. The second constraint indicates that at least one more candidate in the voter's approval set must be present in the deviating committee.  As before,  if we remove the final two integral constraints, the Program can be solved in polynomial time by replacing the objective function with $\sum_{j\in C} x_j$ while adding the constraint $\sum_{i\in [m]} z_i \ge 1$. Suppose after removing the final two constraints, the optimal solution is $\{x_j\}$ and $\{z_i\}$. 

\paragraph{The $O(\log m)$ Approximation.} We scale the variables so that there is some $i$ with $z_i = 1$. This means that $\sum_{j \in C \setminus W} x_j \ge 1$. As a first step, we set $x'_j = \max(x_j, 1/m^2)$ for $j \in C \setminus W$; this will not change the LP objective by more than a constant factor.

The rest of the rounding scheme is as follows:
\begin{enumerate}
    \item Choose $\alpha \in [0,1]$ uniformly at random. 
    \item If $z_i \ge \alpha$, we set $\h z_i = 1$, else we set it to $0$.
    \item For $j \in W$, if $x_j \ge \alpha$, set $\h x_j = 1$, else $\h x_j = 0.$
    \item For $j \notin W$, set $\h x_j = 1$ with probability $\min(1, x_j/\alpha)$, else set $\h x_j = 0$.
    \item If $\h z_i = 1$, check if $\h x_j = 1$ for some $j \in A_i \setminus W$. If so, set $\doublehat{z_i} = 1$, else set it to zero.
\end{enumerate} 

\paragraph{Analysis.} First note that if $\h z_i = 1$, then $\h x_j = 1$ for all $j \in W \cap A_i$. Further, repeating the analysis in \cref{sec:roundproof} shows that $\mathbb{E}[\sum_j \h x_j] = O(\log m) \sum_j x_j$. 

Suppose $\h z_i = 1$. Then $z_i \ge \alpha$. If there is some $j \in C \setminus W$ with $x_j \ge \alpha$, then $\h x_j = 1$ and the second constraint is satisfied for this voter, so that $\doublehat{z}_i = 1$. Otherwise, we can pretend all $x_j$ for $j \in A_i \setminus W$ are first scaled up by $1/\alpha$ and then randomly rounded. After scaling up, we have $\sum_{j \in A_i \setminus W} x_j \ge 1$, so that the randomized rounding chooses sets one of these $\h x_j = 1$ with probability at least $(1-1/e)$. Therefore, we have $\mathbb{E}[\doublehat{z_i}] \ge (1-1/e) z_i$. Proceeding as in \cref{sec:roundproof}, this shows an $O(\log m)$ approximation. 

The proof of the $O(\log n)$ approximation follows similar lines to that in \cref{sec:roundproof}, and is omitted.

\end{document}